\documentclass[a4paper,12pt]{article}

\usepackage{amssymb,amsmath,latexsym,amsthm} \usepackage{xcolor}
\usepackage{enumerate} \usepackage{dsfont} \usepackage{geometry}
\usepackage{graphicx}

\usepackage{stmaryrd} 
\usepackage{authblk} %
\usepackage[showonlyrefs]{mathtools}
\usepackage [normalem]{ulem} %
\usepackage[latin1]{inputenc} %
\usepackage{yfonts}
\usepackage{autonum}

 \newcommand{\SUS}{\subset}

\newcommand{\psum}[1]{\sum_{x\in L}{\vphantom{\sum}}'}

  { \end{itemize} }
  { \end{enumerate} }

\usepackage{array}
\newcolumntype{L}{>{\displaystyle} l <{}}

\newcounter{pcounter} 
 
\def\ED{\end{document}}

\newcommand{\IDEE}[1]{}

\renewcommand{\em}[1]{{\EMMM{#1}}}

\newcommand{\BS}[0]{\backslash}

\newcommand{\HW}[1]{} 

\newcommand{\OL}[0]{\overline} %

\usepackage{marginfix} 
\newcounter{margcount} 
\renewcommand{\t}{\tilde}
\newcommand{\p}{\partial}
\renewcommand{\d}{\delta}

\newcommand{\DD}{{\mathcal D}}
\newcommand{\EE}{{\mathcal E}}
\newcommand{\FF}{{\mathcal F}}

\newcommand{\LL}{{\mathcal L}}

\newcommand{\PP}{{\mathcal P}}

\newcommand{\R}{\mathbb{R}}
\newcommand{\N}{\mathbb{N}}
\newcommand{\Z}{\mathbb{Z}}
\newcommand{\C}{\mathbb{C}}

\newcommand{\alp}{\alpha}

\newcommand{\eps}{\epsilon}

\renewcommand{\phi}{\varphi}

\newcommand{\Lam}{\Lambda}

\DeclareMathOperator*{\Id}{{Id}}

\DeclareMathOperator*{\supp}{supp}

\def \wto{\rightharpoonup}

\renewcommand{\iint}{\int\!\!\!\!\int}

\def\XXint#1#2#3{{\setbox0=\hbox{$#1{#2#3}{\int}$}
\vcenter{\hbox{$#2#3$}}\kern-.5\wd0}}

\newcommand{\upref}[2]{\hspace{-0.8ex}\stackrel{\eqref{#1}}{#2}} 

\newcommand{\lupref}[2]{\hspace{0ex} \stackrel{\eqref{#1}}{#2}} 

\usepackage{color}
\definecolor{verylightblue}{rgb}{0.95, 0.95, 0.95}  
\definecolor{lightblue}{rgb}{0.7, 0.7, 1}
\definecolor{dblue}{rgb}{0.1, 0.1, 0.9}

\definecolor{eqyellow}{rgb}{0.9375,0.8984,0.5469}
\definecolor{subeqyellow}{rgb}{1,0.9373,0.8353}

\definecolor{mygreen}{rgb}{0.3, 0.6, 0.3} 
\definecolor{verylightgreen}{rgb}{0.95, 0.95, 0.95} 
\definecolor{verydarkgreen}{rgb}{0, 0.5, 0}
\definecolor{darkgreen}{rgb}{0.85, 0.85, 0.85}  
\definecolor{mydarkgreen}{rgb}{0, 0.5, 0} 

\definecolor{mybrown}{rgb}{0.85, 0.4, 0.3}
\definecolor{verylightbrown}{rgb}{0.98, 0.72, 0.58}
\definecolor{verydarkbrown}{rgb}{0.44, 0.26, 0.26}

\definecolor{orange}{rgb}{1, 0.5, 0}

\definecolor{BurntOrange}{rgb}{1,0.356,0}

\definecolor{mydarkred}{rgb}{1,0.086,0.255}

\definecolor{RoseVYDP}{rgb}{0.84,0.086,0.255}

\definecolor{dgreen}{rgb}{0, 0.8, 0.5}     

\definecolor{CanaryBRT}{rgb}{1,0.76,0.26}

\definecolor{cyan}{rgb}{0, 1, 1}
\definecolor{verylightgray}{rgb}{0.95, 0.95, 0.95}
\definecolor{verylightgray}{rgb}{0.95, 0.95, 0.95}
\definecolor{verylightred}{rgb}{1, 0.8, 0.78}
\definecolor{verylightyellow}{rgb}{0.99, 0.98, 0.5}

\newcommand{\ignore}[1]{{}}

\newcommand{\NNN}[2]{\|#1\|_{#2}}

\newlength{\bibitemsep}
\newlength{\bibparskip}
\let\oldthebibliography\thebibliography
\renewcommand\thebibliography[1]{%
  \oldthebibliography{#1}%
  \setlength{\parskip}{\bibitemsep}%
  \setlength{\itemsep}{\bibparskip}%
}
\makeatletter \newcommand{\VEC}[2][r]{
  \gdef\@VORNE{1}
  \left(\hskip-\arraycolsep \begin{array}{#1}\vekSp@lten{#2}\end{array}
    \hskip-\arraycolsep\right) } \def\vekSp@lten#1{\xvekSp@lten#1;vekL@stLine;}
\def\vekL@stLine{vekL@stLine} \def\xvekSp@lten#1;{\def\temp{#1}%
  \ifx\temp\vekL@stLine \else \ifnum\@VORNE=1\gdef\@VORNE{0} \else\@arraycr\fi
  #1 \expandafter\xvekSp@lten \fi} \makeatother

\newcommand{\NP}[2]{\NNN{#1}{L^{#2}}}

\newtheorem{theorem}{Theorem}
\newtheorem{definition}[theorem]{Definition}

\newtheorem{proposition}[theorem]{Proposition}
 \newtheorem{lemma}[theorem]{Lemma}
\theoremstyle{definition} \newtheorem{remark}[theorem]{Remark}

 \renewcommand{\t}{\tilde}

\numberwithin{equation}{section}  \def\XXint#1#2#3{{\setbox0=\hbox{$#1{#2#3}{\int}$}
    \vcenter{\hbox{$#2#3$}}\kern-.5\wd0}} 
 \allowdisplaybreaks

\renewcommand{\em}[1]{\underline{#1}}

\usepackage{scalerel,stackengine} \stackMath \newcommand\widecheck[1]{%
  \savestack{\tmpbox}{\stretchto{%
      \scaleto{%
        \scalerel*[\widthof{\ensuremath{#1}}]{\kern-.6pt\bigwedge\kern-.6pt}%
        {\rule[-\textheight/2]{1ex}{\textheight}}
      }{\textheight}%
    }{0.5ex}}%
  \stackon[1pt]{#1}{\scalebox{-1}{\tmpbox}}%
}

\makeatletter \newcommand\avsuminner[2]{%
  {\sbox0{$\m@th#1\sum$}%
    \vphantom{\usebox0}%
    \ooalign{%
      \hidewidth \smash{\vrule height\dimexpr\ht0+1pt\relax
        depth\dimexpr\dp0+1pt\relax}%
      \hidewidth\cr $\m@th#1\sum$\cr }%
  }%
} \makeatother

\author[*]{\rm Laurent B\'{e}termin} \author[**]{\rm Hans Kn\"upfer}
\affil[*]{\normalsize{QMATH, Department of Mathematical Sciences, University of Copenhagen, Universitetsparken 5, DK-2100 Copenhagen \O, Denmark. \texttt{betermin@math.ku.dk}. ORCID id: 0000-0003-4070-3344}}
\affil[**]{Institute of Applied Mathematics and IWR, University of Heidelberg, Im Neuenheimer Feld 205, 69120 Heidelberg, Germany. \texttt{knuepfer@uni-heidelberg.de} }

\begin{document}

\title{Optimal lattice configurations for interacting spatially extended particles
} \date\today

\maketitle

\begin{abstract}
  We investigate lattice energies for radially symmetric, spatially extended particles interacting
  via a radial potential and arranged on the sites of a two-dimensional Bravais
  lattice. We show the global minimality of the triangular lattice among Bravais
  lattices of fixed density in two cases: In the first case, the distribution of
  mass is sufficiently concentrated around the lattice points, and the mass
  concentration depends on the density we have fixed. In the second case, both
  interacting potential and density of the distribution of mass are described by
  completely monotone functions in which case the optimality holds at any fixed
  density.
\end{abstract}
  \setcounter{tocdepth}{2}
\tableofcontents

\section{Introduction and main results}

One key objective in crystallization theory is to understand the optimal
arrangement of particles, interacting with some nonlocal interaction
potential. In particular, it has been shown that the triangular lattice is
optimal among Bravais lattices for a large class of interaction potentials if
one assumes that the particles are located on the lattice sites
\cite{Mont,NonnenVoros,AftBN,Sandier_Serfaty,CheOshita,Betermin:2014fy,BetTheta15}. The
idea of this paper is to generalize these results by considering the optimal
arrangement for particles which are not necessarily concentrated on a single
point, but may be spatially extended. We note that the question for optimal
arrangement of diffuse particles appears naturally in various systems in
condensed matter theory \cite{HeyesBranka} and quantum physics, including
e.g. Thomas-Fermi model \cite{BlancTFLattices} (where the electron density plays
the role of the diffuse particle), diblock copolymer systems in the low volume fraction limit
\cite{OhtaKawasaki} and magnetized disks interactions \cite{Disksinteractions}. For related mathematical analysis for these systems,
we refer refer e.g. to
\cite{Betermin:2014fy,ChoksiPeletier-2008,ChoksiPeletier-2009,KnuepferMuratovNovaga-2016}. We
  also note that systems with localizing and delocalizing interactions, although
  in a dynamic setting, are also relevant in biological models related to
  swarming and flocking, see e.g. \cite{BernoffTopaz-2011,
    TopazBertozzi-2004,BalagueCarrilloEtal-2013,Carrilonowinteraction} and the
  references therein. We show that for particles with radially symmetric mass
distribution, the triangular lattice is still optimal if the mass of each
particle is either sufficiently concentrated near its center, or if the mass
distribution is described by a completely monotone function.

\medskip

We consider a collection of identical particles with center located on the sites
of a two-dimensional Bravais lattice $L \SUS \R^2$. The mass distribution of the
particle at the lattice site $x \in L$ is given by $\mu(\cdot - x)$ for some
given probability measure $\mu \in \PP(\R^2)$. Summing the interaction energy
between different particles over all the lattice sites, we arrive at the
  total energy per particle of the lattice:
\begin{definition}[Lattice Energy] \label{def-energy} %
  For any Bravais lattice $L\subset \R^2$ and $\mu \in \PP(\R^2)$ , we define
  the interaction energy $\EE_{f,\mu}[L]$ by
  \begin{align} \label{energy} %
    \EE_{f,\mu}[L]:=\psum{x\in L} \iint_{\R^2\times \R^2} f(u-x-v) d\mu(u) d\mu(v)
    \quad \in [0,\infty].
  \end{align}
  The interaction potential $f: \R^2 \to \R$ is assumed to be in the class
  $\FF$, defined as the set of functions $f : \R^2 \to \R$ such that there
  exists $C>0$ and $\eta>0$ such that
  $|f(x)|+|\hat{f}(x)|\leq \frac{C}{(1+|x|)^{2+\eta}}$ for any $x\in \R^2$ and
  such that for some nonnegative Radon measure $\mu_f$, we have
  \begin{align} \label{first-con} %
    f(x)=\int_0^{\infty} e^{-|x|^2 t}d\mu_f(t) &&\text{for any $x\in \R^2$}.
    \end{align}
\end{definition}
Here and in the sequel, we use the notation
$\psum{x \in L} := \sum_{x \in L \BS \{ 0 \}}$. The Fourier transform $\hat f$
is defined in \eqref{def-fourier}. We also note that the condition
\eqref{first-con} is equivalent to the fact that $F : (0,\infty) \to \R$, given
by $F(|x|^2) := f(x)$ is completely monotone, i.e. $(-1)^k \p^k F(r) \geq 0$ for
all $k \in \N$ and $r> 0$. In particular, standard potentials such as the
Gaussian potential $g_\alpha(x):=e^{-\alp |x|^2}$, $\alpha>0$ and
$f_a(x)=(a+|x|^2)^{-s}$, $a>0$, $s>1$,
are included into the class of considered potentials (see also
\cite[Sec. 2.3]{BetTheta15} and \cite{ComplMonotonic}). We note that the
  energy can be in general infinite since the measure is not assumed to be
  absolutely continuous.

\medskip

In previous contributions, most results have been concerned with the
case when the particles are located on a single point, i.e. when $\mu$ is given
by a Dirac distribution. In this case, the optimality of the triangular lattice
\begin{align} \label{triangle} %
  \Lam:=\sqrt{\tfrac{2}{\sqrt{3}}}\left[ \Z \left( 1,0 \right) \oplus \Z
  \left(\tfrac{1}{2},\tfrac{\sqrt{3}}{2} \right) \right],
\end{align}
among the class of Bravais lattices with prescribed density has been established
for various interaction potentials, including the class $\FF$. One key result is
the seminal paper by Montgomery \cite{Mont} where the optimality of the
triangular lattice is shown for a class of Gaussian interaction potentials in
which case the lattice energy is described by the lattice theta function
$\theta_L$ given by \eqref{def-theta}. Furthermore, the proof of the minimality
of the triangular lattice for a class of Riesz potentials has been established
in a series of papers from number theorists \cite{Rankin,Cassels,Eno2,Diananda}
in their analysis of the Epstein zeta function.  Montgomery result for the
lattice theta function was used in order to prove the optimality of a triangular
lattice in several physical systems such as Ginzburg-Landau vortices or
Bose-Einstein Condensates in the periodic case
\cite{NonnenVoros,AftBN,NierTheta,Sandier_Serfaty,Betermin:2014rr}. In
\cite{Betermin:2014fy,BetTheta15,Beterloc,BeterminPetrache,Beterminlocal3d}, the
first author has studied this minimization problem among Bravais lattices for
more general potentials $f$. As explained in
  \cite[Sect. 2.5]{Blanc:2015yu}, the minimization of energies per point among
  Bravais lattices is related to the {\it crystallization conjecture}, once the
  crystallization -- i.e. the periodicity of the ground state for this system of
  interacting (extended) particles -- is assumed. Furthermore, even though this
  conjecture has not been proved so far, some interesting advances have been
achieved in the last decade for general configurations of particles in dimension
$d\in \{2,3\}$ (we refer e.g. to
\cite{Crystal,TheilFlatley,Suto1,Suto2,Stef1,Stef2,Luca:2016aa}).

\medskip

In our first result, we consider particles which are described by radially
  symmetric mass distributions $\mu \in \PP(\R^2)$.  Moreover, we define in Section \ref{subsec-local} a metric on the space of Bravais lattices with unit density. The open ball of radius $R$ and centred at $L$ is then denoted by $B_R(L)$. We show that the
triangular lattice is energetically optimal if the particles are sufficiently
concentrated:
\begin{theorem}[Radially symmetric mass distribution]\label{thm-1}
  Let $f\in \FF$, $\mu \in \PP(\R^2)$ be rotationally symmetric with respect to the origin and $\mu_\eps$
  for any measurable set $F \SUS \R^2$ be given by
  \begin{align} \label{def-mueps} %
    \mu_\eps(F) := \mu(\eps F).
  \end{align}
  Furthermore, suppose that $\EE_{f,\mu_\eps}[L] < \infty$ for any $L\in B_c(\Lam)$, for some $c>0$, and for all
  sufficiently small $\eps$, where the triangular lattice $\Lam$ is defined by
  \eqref{triangle}. Then there exists $\eps_0 = \eps_0(f,\mu)$ such that for any
  $0 \leq \eps < \eps_0$, $\Lam$ is the unique minimizer, up to rotation, of
  $\EE_{f,\mu_\eps}$ among all Bravais lattices of unit density.
\end{theorem}
Our next result shows that if the mass distribution of each particle can be
expressed by a completely monotone function in terms of the square root of the
radius, then the triangular lattice is again the global minimizer:
  \begin{theorem}[Completely monotone mass distribution] \label{thm-2} %
    Let $f \in \FF$ and suppose that $\mu \in \PP(\R^2)$ satisfies
    $d\mu(x)=\rho(x)dx$ for some $\rho:\R^2\to \R$ satisfying
  \begin{align} \label{first-con2} %
    \rho(x)=\int_0^{\infty} e^{-|x|^2 t}d\mu_\rho(t) &&\text{for any
                                                        $x\in \R^2$},
    \end{align}
    for some nonnegative Radon measure $\mu_\rho$.  Then the triangular
    lattice $\Lam$ is the unique minimizer, up to rotation, of $\EE_{f,\mu}$
        among all Bravais lattices of unit density.
\end{theorem}
The argument for Theorem \ref{thm-1} is based on a perturbation argument and the
local optimality of the triangular lattice for $\eps $ sufficiently small. The
proof of Theorem \ref{thm-1} is given in Section \ref{proof-thm1}.  The proof of
Theorem \ref{thm-2}, given in Section \ref{subsec-proofth1}, is a generalization
of the arguments used in the proof of the optimality of the triangular lattice
for point masses. In particular, the argument is based on the fact that the
triangular lattice is self-dual and that the class of completely monotone
functions is stable with respect to multiplication. Some of the tools in
\cite{Betermin:2014fy,BetTheta15,Beterloc} will be used in order to prove our
results as well as computations from
\cite{SarStromb,Coulangeon:2010uq}. We note that if $\mu$ satisfies the
  condition of Theorem \ref{thm-2}, then for any $\eps > 0$, the density
  $\mu_\eps$ defined by \eqref{def-mueps} also satisfies these conditions.  In
  this sense, the result of Theorem \ref{thm-2} is independent of the
  concentration of the particles, while this is not the case for Theorem
  \ref{thm-1}.

\medskip

The paper is concerned with the case of absolutely summable interaction
potentials in two dimensions. However, we believe that these results can be
generalized. In particular, the results should still hold true for
non--integrable interaction potentials such as the Riesz potentials
$f(x) = |x|^{-s}$ for any $s > 0$ if a cut-off argument near the origin is
applied and the Ewald summation method is used to obtain summability for large
$x$ in the non-summable case (see for instance \cite[Section
IV]{SaffLongRange}). Moreover, as explained in Remark \ref{rem-highdim}, all our
results related to local minimality could be written in dimensions $4,8$ and
$24$ where the densest packing is self-dual, and is expected to be the unique
minimizer of the theta function for any $\alpha>0$. In these dimensions, only
local minimality results have been proved \cite{SarStromb,Coulangeon:2010uq},
but some recent results \cite{Viazovska,CKMRV} have opened the door for a proof
of this universality in dimension $8$ and $24$, and then to a generalization of
our result in those dimensions (see also Remark \ref{rem-highdim}). Another
interesting question would be to consider the Lennard-Jones-type potentials of
the form $f(x)=a_1|x|^{-p}-a_2|x|^{-q}$, $p>q$, which is widely used in
molecular simulation (see \cite[Section 6.3]{BetTheta15} for some examples). In
this situations, the generalization of our results should hold under the
additional assumption that the density of the lattice is sufficiently high (see
\cite{Betermin:2014fy,BetTheta15} for the optimality of the triangular lattice
in the point mass case for Lennard-Jones-type potentials). Another interesting
extension would be to allow for signed measures
or to consider lattice configurations with alternating charges as in
\cite{BeterminKnuepfer-preprint}.

\paragraph{Notation.} We recall (see e.g. \cite{Rankin,Mont,Beterloc}) that a
two-dimensional Bravais lattice $L \SUS \R^2$ is a set of the form
$L=\Z u_1 \oplus \Z u_2$ for two linearly independent vectors $u_1$,
$u_2 \in \R^2$. The dual lattice $L^*$ is given by $L^*$ $=\{ p \in \R^2 \ $:
$\ p\cdot x\in \Z \text{ for all $x\in L$} \}$. Up to isometry, any Bravais
lattice $L\subset \R^2$ of unit density can be written as
\begin{align} \label{of-form} %
  L := \OL L(x,y) := \Z\left(\frac{1}{\sqrt{y}},0\right)\oplus \Z\left(\frac{x}{\sqrt{y}},\sqrt{y}\right) %
  && \text{with $(x,y)\in D$,}
\end{align}
where $(x,y) \in D$ are uniquely determined in the fundamental domain $D$ with
\begin{align}
  D:=\left\{(x,y)\in \big[0,\tfrac{1}{2}\big]\times (0,\infty) : x^2+y^2\geq 1\right\}.
\end{align}
The set $\DD$ is the set of lattices of the form \eqref{of-form}. It
is hence sufficient to restrict our analysis to the set of lattices in
$\DD$.
Finally, by $\PP(\R^2)$ we denote the space of probability measures. The
  Laplace transform of a Radon measure $\mu$ on $(0,\infty)$ is
\begin{align}
    (\LL \mu)(r) := \int_0^\infty e^{-rt} d\mu(t).
\end{align}
The Fourier transform $\hat f : \R^2 \to \C$ is denoted by
\begin{align} \label{def-fourier} %
  \widehat{f}(p) := \int_{\R^2} f(x) e^{-2i\pi  p\cdot x }dx.
\end{align}  

\section{Proofs}\label{sec-proofs}

\subsection{Preliminaries}\label{subsec-local}

We will work with Bravais lattices of the form \eqref{of-form}. To get a
  notion of local minimality, we introduce a topology and define the distance
of two lattices by
\begin{align} \label{metric} %
  d(L_1,L_2) = \sqrt{\inf_{k \in \Z} |x_1 - x_2 - \tfrac 12 k|^2 + |y_1 - y_2|^2} && \text{for  $L_i = \OL L(x_i,y_i)\in \DD$}.
\end{align}
We also denote by $B_R(L)$ the set of all $\t L\in \DD$ with $d(L,\t L)<R$. We
write $f \in C^k(\DD)$ if $f$ is $k$-times differentiable in $(x,y)$. As usual
we denote the gradient by $\nabla E[L]:= (\partial_x E,\partial_y E)[L]$. The
Hessian $D^2 E$ is defined as the matrix of second derivatives w.r.t. $x$,
$y$. $D^3 E$ is correspondingly the tensor of all third derivatives. We also
define:
  \begin{definition}[Local minimizers and critical points in $\DD$] %
    Let $E : \DD \to \R$.  $L$ is a strict local minimizer in $\DD$ if there
      is $\eta > 0$ such that $E[L] < E[\t L]$ for all $\t L\in B_\eta(L)$.
      Furthermore, $L$ is a critical point of $E \in C^1(\DD)$ in $\DD$ if
      $\nabla E[L] = (0,0)$. 
\end{definition}
Before we consider the case of diffuse charge configurations, we collect some
basic facts about the interaction energy for point charges:
\begin{definition}[Sharp interface energy] %
  For $h \in \FF$ and $L\in \DD$, we define
  \begin{align} \label{E-point} %
    E_h[L]:=\sum_{p\in L}{\vphantom{\sum}}'  h(p).
  \end{align}
\end{definition}

We turn to the investigation of local minimality of lattice energies for spread
out particles. We first recall the summation formula of Poisson \cite[Cor.
  2.6]{SteinWeiss}:
\begin{proposition}[Poisson summation formula] \label{prp-poisson} %
  For $f\in \FF$ and $L\in \DD$, we have
  \begin{align} \label{eq-poisson} %
    \sum_{x\in L} f(x+z)=\sum_{p\in L^*} e^{2\pi i p\cdot z}\hat{f}(p) &&\text{for all $z\in \R^2$,}
  \end{align}
  and the series on both sides of \eqref{eq-poisson} converge absolutely.
\end{proposition}
We note that the assumption in Definition \ref{def-energy} on the decay of $f$
and $\hat f$ is included to ensure that \eqref{eq-poisson} can be applied and to
get absolute convergence.

\medskip

A representation of the energy in terms of Fourier variables is given next:
\begin{proposition}[Fourier representation of $\EE_{f,\mu}$]
  \label{prp-epsfou} %
  Let $f\in \FF$, $\mu \in \PP(\R^2)$ be rotationally symmetric with respect to the origin and $L\in \DD$ be
  such that $\EE_{f,\mu}[L]<\infty$. Then
\begin{equation}
\EE_{f,\mu}[L]=E_h[L^*]+\hat{f}(0)-(f\ast \mu \ast \mu)(0)
\end{equation}  
where $E_h$ is defined in \eqref{E-point} and where
  \begin{align}\label{def-g} %
    \quad h(p) :=\hat{f}(p)g(|p|)^2, &&
                                        g(t)
                                        := \int_0^\infty J_0(2\pi s t)  d\psi(s).
  \end{align}
  Here, $\psi$ is the Lebesgue--Stieltjes measure of $t\mapsto \mu(B_t)$, i.e.
  $\psi([r_1,r_2)) = \mu(B_{r_2}) - \mu(B_{r_1})$ and $J_0$ is the Bessel
  function of the first kind.
\end{proposition}
\begin{proof} %
  By Fubini and Poisson's summation formula (Proposition \ref{prp-poisson}), we
  get
  \begin{align}
    \EE_{f,\mu}[L]&=\iint_{\R^2\times \R^2} \sum_{p\in L^*}{\vphantom{\sum}}'  \hat{f}(p) e^{2i \pi p\cdot (u-v)}d\mu(u)d\mu(v)+\hat{f}(0)-(f\ast \mu \ast \mu)(0) \\
  &= \sum_{p\in L^*}{\vphantom{\sum}}' 
    \hat{f}(p)\Big(\int_{\R^2}e^{2i\pi p\cdot u} d\mu(u) \Big)
    \Big(\int_{\R^2}e^{-2i\pi p\cdot v} d\mu(v) \Big) +\hat{f}(0)-(f\ast \mu \ast \mu)(0) \\%
    &=\sum_{p\in L^*}{\vphantom{\sum}}' \hat{f}(p) \hat{\mu}(p)^2 +\hat{f}(0)-(f\ast \mu \ast \mu)(0).
  \end{align}
  The proof is completed by noting that $\hat{\mu}$ is given by the
    Hankel-Stieltjes transform (\cite[Section 2]{ConnetSchwartz}) and hence
    $\hat \mu = g$ where $g$ is given in \eqref{def-g}.
\end{proof}
Since $\Lam^*=\Lam$ up to rotation, a consequence of Proposition
\ref{prp-epsfou} is that proving the (local) minimality of $\Lam$ for
$\EE_{f,\mu}$ in $\DD$ is the same as proving it for $E_{h}$. In particular, we
can use the following result:
\begin{lemma}[\cite{Coulangeon:2010uq,Beterloc}] \label{lem-critictriang} Let
  $L=\Lam$ be the triangular lattice of unit density. Let
  $f\in L^1(\R^2)\cap C^2(\R^2\backslash\{0\})$ be rotationally symmetric with respect to the origin and $F$ be given by
  $F(|x|^2) := f(x)$. Then
  \begin{enumerate}
  \item $\Lam$ is a critical point of $E_f$ in $\DD$.
  \item We have $D^2 E_f[\Lam] = T_f \Id_{\R^2}$ where
    $\Id_{\R^2}$ is the identity on $\R^2$ and
    \begin{align}
      T_f:=\frac{4}{\sqrt{3}}\sum_{m,n}{\vphantom{\sum}}'n^2 F'\big(\tfrac{2}{\sqrt{3}}(m^2+mn+n^2) \big)+\frac{4}{3}\sum_{m,n}{\vphantom{\sum}}'n^4 F''\big(\tfrac{2}{\sqrt{3}}(m^2+mn+n^2)\big).
    \end{align}
  \end{enumerate}
\end{lemma}
We next give a generalization of \cite[Thm. 4.6]{Coulangeon:2010uq} for the
two-dimensional case.
\begin{proposition}[Local optimality of $\Lam$ for $E_f$]\label{prop-locmin}
  Let $f\in \FF$. Then $\Lam$ is a strict local minimum of $E_f$ in $\DD$
  with positive second derivative.
\end{proposition}
\begin{proof}
  We first claim that for any $t>0$, the triangular lattice $\Lam$ is a strict
  local minimum in $\DD$ with positive second derivative of the lattice theta
  function
  \begin{align}\label{def-theta}
    \theta_L(t) := \sum_{x\in L} e^{-\pi t|x|^2} = E_{G_t}[L] +1,
  \end{align}
  where $G_t$ is the Gaussian potential $G_t(x)=e^{-\pi t |x|^2}$. As explained
  in \cite{Coulangeon:2010uq} for dimensions $d\in \{4,8,24\}$, a direct
  consequence of \cite[Eq. (46)]{SarStromb} is that
  $D^2\hspace{1mm}E_{G_t}[\Lam]$ is positive definite for any $t>0$.  Since we
  can write, by Fubini's theorem,
  \begin{align}
    E_f[L]= \int_0^{\infty}\left( \theta_L(t/\pi)-1\right)d\mu_f(t),
  \end{align}
  for any $f \in \FF$, the same result also holds all $f \in \FF$.
\end{proof}
Finally, it is useful to note that the class of functions $\FF$
is stable under the application of the Fourier transform:
\begin{lemma}\label{lem-ffhat}
  We have $f\in \FF $ if and only if $\hat{f}\in \FF $.
\end{lemma}
\begin{proof}
  We need to show that $\hat{f}$ can be written as in \eqref{first-con}. By
  definition, for any $f\in \FF$ there is a nonnegative Radon measure $\mu_f$
  such that \eqref{first-con} holds.  By application of the Fourier transform
  and by Fubini's theorem, we then have
  \begin{align}
    \hat{f}(p) %
    &=\int_0^{\infty} \left( \int_{\R^2} e^{-|x|^2 t} e^{-2 i \pi x\cdot p} dx \right)d\mu_f(t) %
      =\int_0^{\infty} e^{-\frac{\pi^2}{t} |p|^2}t^{-1}d\mu_f(t). 
  \end{align}
  It follows that $\hat f$ can be expressed in the form \eqref{first-con} with
  $\mu_f$ replaced by some nonnegative Radon measure $\mu_{\hat{f}}$.
\end{proof}
The proof of the next result about local minimality of $\Lam$ for
$\EE_{f,\mu_\eps}$ and small $\eps$ is given in the Appendix:
\begin{proposition}[Local minimality of $\Lam$ for small
  $\eps$]\label{prp-localmin}
  Let
  $f\in \FF$ and $\mu\in \PP(\R^2)$ be rotationally symmetric with respect to the origin. With $\mu_\eps$ given
  by \eqref{def-mueps}, we assume that $\EE_{f,\mu_\eps}[L]<\infty$ for any $L\in B_c(\Lam)$, for some $c>0$, and for all
  sufficiently small $\eps$. Then there exists $\eps_1=\eps_1(f,\mu)$ and
  $c_0 = c_0(f,\mu)>0$ (independent of $\eps$) such that for any
  $0\leq \eps <\eps_1$, the triangular lattice $\Lam$ is a strict local
  minimizer of $\EE_{f,\mu_\eps}$ in the set of lattices $B_{c_0}(\Lam)$.
\end{proposition}
\begin{proof}
By Proposition
  \ref{prp-epsfou} and the fact that $\Lam^*=\Lam$, up to rotation, we need to
  show that $\Lam$ is a strict minimizer of $E_{h_\eps}$, where
  $h_\eps(p) :=\hat{f}(p)g_\eps(|p|)^2$ and $g_\eps(r) :=g(\eps r)$ with $g$
  defined in \eqref{def-g}. We furthermore set
  $H_\eps(|x|^2) := h_{\eps}(x) = \Phi(|x|^2)G_\eps(|x|^2)^2$ with
  $\Phi(|p|^2) :=\hat{f}(p)$, $G_\eps(|p|^2) :=g_\eps(|p|)$, i.e.
\begin{align}
  G_\eps(r)=\int_0^{\infty} J_0(2\pi s \eps \sqrt{r})d\psi(s).
\end{align}
By Lemma \ref{lem-critictriang} and
the fact that $h_\eps$ is a radial function, $\Lam$ is a critical point of
$E_{h_\eps}$ in $\DD$ for any $\eps>0$ and a condition for the strict positivity
of its Hessian at $L=\Lam$ is
\begin{align}
  T_{h_\eps}:=\frac{4}{\sqrt{3}}\sum_{m,n}{\vphantom{\sum}}' n^2 H_\eps'\Big(\tfrac{2}{\sqrt{3}}(m^2+mn+n^2) \Big)+\frac{4}{3}\sum_{m,n}{\vphantom{\sum}}' n^4 H_\eps''\Big(\tfrac{2}{\sqrt{3}}(m^2+mn+n^2)\Big)>0.
\end{align} 
Since $J_0'(r)=-J_1(r)$ and $J_1'(r)=\frac{1}{2}(J_0(r)-J_2(r))$, a simple
calculation yields
\begin{align}
  (G_\eps^2)'(r) =-\frac{2\pi\eps}{\sqrt{r}}\Big(\int_0^{\infty} s J_1( \t r)d\psi(s)\Big)\Big(\int_0^{\infty} J_0(\t r) d\psi(s)\Big),
\end{align}
with the notation $\t r := 2\pi s \eps \sqrt{r}$. Furthermore,
\begin{align}
  (G_\eps^2)''(r)=&\frac{\pi \eps}{r^{3/2}}\int_0^{\infty} J_0(\t r)d\psi(s)\int_0^{\infty} s J_1(\t r) d\psi(s)  %
                    +\frac{2\pi^2 \eps^2}{r}\Big(\int_0^{\infty} s J_1(\t r) d\psi(s)   \Big)^2
  \\
                  &+\frac{\pi^2 \eps^2}{r}\int_0^{\infty} J_0(\t r)d\psi(s)\int_0^{\infty} s^2 (J_2(\t r) - J_0(\t r)) d\psi(s).
\end{align}
We also recall the series expansion of Bessel functions (see e.g. \cite[Eq. (9.1.10)]{AbraStegun})
\begin{gather}
  J_0(\t r)=1+\sum_{m=1}^{\infty} \frac{(-1)^m (\frac{\t r}2)^{2m}}{(m!)^2}, 
  J_1(\t r)=\sum_{m=1}^{\infty} \frac{(-1)^m (\frac{\t r}2)^{2m+1}}{m! (m+1)!}, 
  J_2(\t r)=\sum_{m=1}^{\infty} \frac{(-1)^m (\frac{\t r}2)^{2m+2}}{m!
   (m+2)!},
\end{gather}
\vspace{-4ex}\\
where $\t r = 2\pi s \eps \sqrt{r}$ as before. Thus, using the Leibniz rule
  on $H_\eps''=(\Phi G_\eps ^2)''$, since $\int_0^{\infty} d\psi(s)=1$ and
by substituting the expressions of $J_n(2\pi s\eps \sqrt{r})$, $n\in \{0,1,2\}$,
previously computed and summing on $\Lam$, we obtain
$T_{h_\eps}=T_{\hat{f}}+R_\eps$, where $R_\eps\to 0$ as $\eps\to 0$. By Lemma
\ref{lem-ffhat} we have $\hat{f}\in \FF$, and it follows from Proposition
\ref{prop-locmin} that $T_{\hat{f}}>0$. Therefore, there exists $\eps_1$ such
that for any $0\leq \eps <\eps_1$, $T_{h_\eps}>0$, and $\Lam$ is a local minimum
of $E_{h_\eps}$ in $\DD$, for any such $\eps$. Moreover, as all the $J_n$ are
bounded, it is straightforward to prove that for any $L\in B_c(\Lam)$ and any
$0\leq \eps\leq 1$, we have $|D^3 \EE_{h_\eps}[L]|\leq C$ for some $C$
independent of $\eps$. This concludes the proof.
\end{proof}

\subsection{Proof of Theorem \ref{thm-1}}\label{proof-thm1}

In this section, we give the proof of Theorem \ref{thm-1}, related to the case
of small $\eps$. At the end of the section, we give some remarks about the case
of large $\eps$.

\medskip

Before proving Theorem \ref{thm-1}, we show in the following result that a
minimizer $L_0\in \DD$ of $E_{h_\eps}$ is necessarily at a bounded distance from
$\Lambda$.
\begin{lemma}\label{lem-compactlattice}
  Let $f\in \FF$ and $\mu\in \PP(\R^2)$ be rotationally symmetric with respect to the origin. Suppose that
  $\EE_{f,\mu_\eps}[L]<\infty$ for any $L\in B_c(\Lam)$, for some $c>0$, and for all $\eps < \eps_0$ and some $\eps_0 > 0$ where $\mu_\eps$ is
  given by \eqref{def-mueps}. Then there exists $c_1=c_1(f,\mu)>0$ and $\eps_1=\eps_1(f,\mu)$ such that for
  any $0\leq \eps < \eps_1$ and any minimizer $L_0$ of $E_{h_\eps}$ in $\DD$,
  we have $L_0\in B_{c_1}(\Lam)$.
  \end{lemma}
  \begin{proof}
    By Proposition \ref{prp-epsfou}, it suffices to prove the equivalent
      statement for the energy
      \begin{align}
        E_{h_\eps}[L]= \sum_{p\in L}{\vphantom{\sum}}' h_\eps(p) = \sum_{p\in L}{\vphantom{\sum}}' \Phi(|p|^2)\int_0^{\infty} J_0(2\pi s\eps |p|)d\psi(s),
      \end{align}
      where $\Phi(|x|^2)=\hat{f}(x)$ for any $x\in \R^2$.  Note that
      $\mu_{\hat f} \geq 0$ since $f \in \FF$ and hence $\hat f \in \FF$ by
      Lemma \ref{lem-ffhat}.  We recall that, by definition of $\FF$ and Lemma
      \ref{lem-ffhat}, $\Phi$ is a completely monotone function -- and then
      decreasing -- and $\Phi(|x|^2)=\hat{f}(x)\leq C(1+|x|)^{-2-\eta}$ for some
      $\eta>0$. We also note that $\psi([0,\infty))=1$.  Let
      $L=\bar{L}(x,y)\in \DD$ with $(x,y)\in D$ such that
      $E_{h_\eps}[L]<\infty$. Then we have, for any
      $p=m (\frac{1}{\sqrt{y}},0)+n(\frac{x}{\sqrt{y}},\sqrt{y})\in L$ for some
      $(m,n)\in \Z^2$, $|p|^2=\frac{(m+xn)^2}{y}+yn^2$.  We will use that
      $J_0(r) \geq \frac{1}{2}$ for $r \leq 1$ and $J_0(r) \geq - \frac{1}{2}$
      for $r \geq 1$ (which can be checked easily).  Hence, we get, for any
      $0\leq \eps< \eps_0$,
      \begin{align}
     E_{h_\eps}[L] %
      &\geq \frac{1}{2}\sum_{p\in L}{\vphantom{\sum}}' \Phi(|p|^2)   \int_{\{s \leq \frac{1}{2\pi \eps |p|} \}} d\psi(s)-\frac{1}{2}\sum_{p\in L}{\vphantom{\sum}}'\Phi(|p|^2)  \int_{\{s \geq \frac{1}{2\pi \eps |p|} \}} d\psi(s) \\
      & \geq \frac 12 \Big(I_1(\eps;L) - I_2(\eps;L) \Big),
      \end{align}
      where for $L = \OL L(x,y)$ we write
      \begin{align}
        I_1(\eps;L):=\sum_{p\in L}{\vphantom{\sum}}' \Phi(|p|^2)   \psi\Big(\{s\leq \frac{1}{2\pi \eps |p|}  \}\Big), %
        && I_2(\eps;L):=\sum_{p\in L}{\vphantom{\sum}}' \Phi(|p|^2)  \psi\Big(\{s\geq \frac{1}{2\pi \eps |p|}  \}\Big).
      \end{align}
      Let $\eps_2$ be such that
$\psi(\{s\leq \frac{1}{2\pi \eps_2}\}) \geq \frac 34$, and we set $\eps_1=\min(\eps_2,\eps_0)$. Let $I_1 = \sum_p I_1^p$ and $I_2=\sum_p I_2^p$. We first note that
      $I_1^p \geq 0$ and $I_2^p \geq 0$ for all $p \in L$. Since
      $\Phi(\frac 1z) \to \NP{f}{1}$ and $\psi(z) \to 1$ for $z \to \infty$, we
      have $I_1(\eps;L)\geq I_1(\eps_1;L)$ and
      \begin{align} \label{I1-lim} %
        \liminf_{y\to \infty} I_1(\eps_1;L) %
        &\geq \sum_{m\in \Z}{\vphantom{\sum}}' \liminf_{y\to \infty}\Phi\left( \frac{m^2}{y} \right)\psi\Big(\{s\leq
          \frac{\sqrt{y}}{2\pi \eps_1 |m|} \}\Big)=\infty,
\end{align}
uniformly in $\eps$ for all $0\leq \eps < \eps_1$.  By definition of $\eps_2$ we have, for any $0\leq \eps <\eps_1$,
$I_1^p \geq 2 I_2^p$ and hence $I_1^p - I_2^p > \frac 12 I_1^p$ if
$|p| \leq 1$. Together with \eqref{I1-lim} this implies that
$I_1 - \sum_{|p| \leq 1} I_2^p \to \infty$ as $y\to \infty$. It remains to
estimate $\sum_{|p| \geq 1} I_2^p$. We first consider the terms with $n = 0$,
i.e. $p = (\frac{m}{\sqrt{y}},0)$ with $|p|^2 = \frac{m^2}y \geq 1$.  Since
$\psi \leq 1$, we have for any $y\geq 1$ and any $0\leq \eps<\eps_1$,
\begin{align}
  \sum_{|p| > 1, p = (\frac{m}{\sqrt{y}},0)} I_2^p(\eps,L) \leq \sum_{|m| > 1, m \in \Z} \Phi(m^2) %
  \leq \sum_{|m|>1,m\in \Z}\frac{1}{(1+|m|)^{2+\eta}} \leq C.
\end{align}
We turn to the estimate of the $n\neq 0$ terms. For any $y\geq 1$,
$x\in \left[0,\frac{1}{2}\right]$ and $(m,n)\in \Z\times \Z\BS \{0\}$ we have
$\frac{2xmn}{y}\geq -\frac{x}{y}(m^2+n^2)$, $x^2-x\geq -\frac{1}{4}$,
$y-\frac{1}{4y}> \frac{y}{2}$ and
\begin{align}  %
  \frac{(m+xn)^2}{y}+yn^2 %
  &\geq \frac{(1-x)m^2}{y}+\frac{(x^2-x+y^2)n^2}{y}\geq
  \frac{m^2}{2y}+\left( y-\frac{1}{4y} \right)n^2 \notag \\
  & >\frac{m^2}{2y}+ \frac{yn^2}{2}. \label{dazuu}
  \end{align}
  Furthermore, we recall that $\Phi$ is decreasing. Therefore, using
  $\psi\leq 1$ we get, for any $y\geq 1$, any $x\in \left[0,\frac{1}{2} \right]$
  and any $0\leq \eps < \eps_1$,
 \begin{align}
   \hspace{6ex} & \hspace{-6ex} %
                  \sum_{|p| > 1, n \neq 0} I_2^p(\eps,L)\leq \sum_{m,n : n \neq 0}\Phi\Big( \frac{(m+xn)^2}{y}+yn^2 \Big)\lupref{dazuu}\leq \sum_{m,n : n \neq 0}\Phi\Big(\frac{1}{2}\Big( \frac{m^2}{y}+yn^2 \Big)\Big)\\
                                         &\leq C\iint_{\R^2}\frac{dudv}{\left(1+\sqrt{\frac{u^2}{2y}+\frac{y}{2}v^2}  \right)^{2+\eta}}= 2C\iint_{\R^2}\frac{dwdz}{\left(1+\sqrt{w^2+z^2}  \right)^{2+\eta}} < \infty,
\end{align}
by the change of variables $w=\frac{u}{\sqrt{2y}}$ and
$z=\frac{v\sqrt{y}}{\sqrt{2}}$. We finally obtain that, for any $0\leq \eps < \eps_1$, any $y\geq 1$ and any $0\leq x\leq 1/2$,
$E_{h_\eps}[L]\geq C(y)$ for some function $C$ that depends only on $f,\mu,$ and
$\eps_1$, but is independent of $\eps \in [0,\eps_1)$, with $C(y) \to \infty$ for
$y \to \infty$. Therefore, there exists $c_1>0$ such that $L_0\in B_{c_1}(\Lam)$
for such minimizer $L_0$.
\end{proof}

\begin{proof}[Proof of Theorem \ref{thm-1}]
  By Proposition \ref{prp-epsfou},
  proving the optimality of $\Lam$ in $\DD$ for $\EE_{f,\mu_\eps}$ is
  equivalent with proving the same for
  \begin{align}
    E_{h_\eps}[L]= \sum_{p\in L}{\vphantom{\sum}}' h_\eps(p) =  \sum_{p\in L}{\vphantom{\sum}}' \hat{f}(p)g_\eps(|p|)^2,
  \end{align}
  where $g_\eps(r)=g(\eps r)$, $g$ is defined by \eqref{def-g} and
  $h_\eps(p) =\hat{f}(p)g_\eps(|p|)^2$. We note that in general $H_\eps$,
  where $H_\eps(|p|^2) := h_\eps(p)$, is not completely monotone so that the
  optimality of the triangular lattice is not clear. However, in the limit
  $\eps \to 0$, $H_\eps$ approaches the completely monotone function $\Phi$,
  defined by $\Phi(|x|^2) := \hat f(x)$. By Lemma \ref{lem-ffhat}, there is
  hence a nonegative Radon measure $\mu_{\hat{f}}$ such that
  $\Phi = \LL \mu_{\hat{f}}$.

By Lemma \ref{lem-compactlattice}, if $L_0$ is a minimizer of $E_{h_\eps}$ for small enough $\eps$, then $L_0\in B_{c_1}(\Lam)$ for some $c_1>0$. In the sequel, we will consider only lattices $L\in B_{c_1}(\Lam)$ such
  that $E_{h_\eps}[L]<\infty$ (otherwise $L$ is not a minimizer).

  We hence approximate $H_\eps$ by a sequence of
  completely monotone functions $\t H_\eps$. More precisely, we construct
  $\t H_\eps$ such that $E_{\t h_\eps}[L]$, defined in \eqref{E-point} for
  $\t h_\eps(p) :=\tilde{H}_\eps(|p|^2)$ satisfies for any
  $L \in B_{c_1}(\Lam)$ 
  \begin{align}
    \left| E_{\t h_\eps}[L] - E_{h_\eps}[L] \right| %
    &\leq C \eps^2 &&\text{for $\eps \to 0$}, \label{cl-1} \\
    E_{\t h_\eps}[L] - E_{\t h_\eps}[\Lam] &\geq C d(L,\Lam)^2. \label{cl-2}
  \end{align}
  We claim that the assertion of the proposition follows from
  \eqref{cl-1}--\eqref{cl-2}. Indeed, if \eqref{cl-1}--\eqref{cl-2} hold, then
  we can estimate for any $L \in B_{c_1}(\Lam)$ ,
  \begin{align}
    E_{h_\eps}[L] - E_{h_\eps}[\Lam] %
    &= E_{h_\eps - \t h_\eps}[L] + \big( E_{\t h_\eps}[L] - E_{\t h_\eps}[\Lam] \big) - E_{h_\eps - \t h_\eps}[\Lam] \\
    &\upref{cl-2}\geq C d(L,\Lam)^2 + E_{h_\eps - \t h_\eps}[L] -  E_{h_\eps - \t h_\eps}[\Lam]  \\
    &\upref{cl-1}\geq C d(L,\Lam)^2 - 2C \eps^2.
  \end{align}
  For any lattice $L$ with
  $E_{h_\eps}[L] \leq E_{h_\eps}[\Lam]$, this implies $d(L,\Lam) \leq C
  \eps$. We claim
  \begin{align} \label{loc-min} %
    E_{h_\eps}[L] \geq E_{h_\eps}[\Lam] &&\text{for all $L\in B_{c_0}(\Lam)$ and all
                                                 $0\leq \eps <\eps_0$,}
  \end{align}
  for some $\eps_0 > 0$ and some $c_0 > 0$. This follows from Proposition
  \ref{prp-localmin}.  Assuming that \eqref{loc-min} holds, we get the statement
  of the theorem.

  \medskip

It remains to show
    that \eqref{cl-1}--\eqref{cl-2} hold.  In order to construct $\t H_\eps$, we
    first approximate $g_\eps^2$ by some function $\t g_\eps^2$. We choose
    $\t \nu \in C_c^\infty((\frac 12, 2))$ such that $\t \nu \geq 0$ and
  $\int_0^\infty \t \nu(t) \ dt = 1$. With the definition
  \begin{align}
    \t g(r)^2 := \LL \t \nu(r) %
    = \int_0^\infty e^{-rt} \t \nu(t) \ dt
  \end{align}
  and $\t g_\eps(r) = \t g(\eps r)$ we then have $\t g^2 \in \FF$,
  $\tilde{g}(0)^2 = 1$ and $\| (\t g^2)' \|_{\infty} \leq C$.  Since also
  $\|(g^2)'\|_{\infty} \leq C$ and $g(0)^2 = 1$, we hence get, for any $r> 0$,
\begin{align} \label{g-close} %
  |g_\eps(r)^2 - \t g_\eps(r)^2| %
  = |g(\eps r)^2 - \t g(\eps r)^2| \leq C \min \{
  \eps^2 r, 1 \}.
  \end{align}
 In terms of
  $\t \nu_\eps(t) := \frac{1}{\eps}\tilde{\nu}\big( \frac{t}{\eps} \big)$, we
  get
  \begin{align}
    \t g_\eps(r)^2 %
    = \tilde{g}(\eps r)^2
    = \int_0^\infty e^{-\eps rt} \t \nu(t) \ dt %
    = \int_0^\infty e^{-rt} \t \nu_\eps(t) \ dt. %
  \end{align}
  We extend $\t \nu$ and $\mu_{\t f}$ by $0$ to measures on $\R$ and define
  \begin{align}
    \t \mu_\eps(t) %
    := (\mu_{\hat{f}}* \t \nu_\eps)(t) = \int_0^\infty \tilde{\nu}_\eps(t-s) d\mu_{\hat{f}}(s) &&\text{for $t > 0$.}
  \end{align}
  In particular, $\mu_{\hat f} \geq 0$. With $u=t-s$ and since
    $\supp(\t \nu) \in \left( \frac{1}{2},2 \right)$ we calculate
  \begin{align}
  \LL \t \mu_\eps(r)&=\int_0^\infty \int_0^\infty  e^{-rt} \tilde{\nu}_\eps(t-s) d\mu_{\hat{f}}(s) dt %
  = \int_0^{\infty}\left( \int_{-s}^{\infty} e^{-r(s+u)} \tilde{\nu_\eps}(u)d u \right)d\mu_{\hat{f}}(s)\\
  &=\int_0^{\infty} e^{-rs} \left( \int_{0}^{\infty} e^{-ru} \tilde{\nu_\eps}(u)d u  \right)d\mu_{\hat{f}}(s) %
    =  \tilde{g}_\eps(r)^2 \int_0^{\infty} e^{-rs}  d\mu_{\hat{f}}(s) \\
    &= \Phi(r) \tilde{g}_\eps(r)^2=: \tilde{H}_\eps(r).
  \end{align}
  In view of \eqref{g-close}, we then get, for any $p \neq 0$,
  \begin{align} \label{g-close2} %
    |H_\eps(|p|^2) - \t H_\eps(|p|^2)| %
    = |\hat f(p)| |g_\eps(|p|)^2 - \t g_\eps(|p|)^2| \leq C |\hat f(p)|
    \min \{ \eps^2 |p|, 1 \}.
  \end{align}
This implies \eqref{cl-1} by the dominated
    convergence theorem since $\hat{f} \in L^1(\R^2)$.

  \medskip

  We turn to the proof of \eqref{cl-2}. We note that, since
  $\t H_\eps=\LL \tilde{\mu}_\eps$,
  \begin{align}
    E_{\t h_\eps}[L] - E_{\t h_\eps}[\Lam] = \int_0^\infty \big[\theta_L(\pi t) - \theta_{\Lam}(\pi t) \big] \t \mu_\eps(t) \ dt,
  \end{align}
  where $\theta_L$ is the lattice theta function given by
  \eqref{def-theta}. From the construction, we also have $\t \nu_\eps \wto \d_0$
  and hence $\t \mu_\eps = \mu_{\hat{f}} *\t \nu_\eps \wto \mu_{\hat{f}}$ in the
  sense of measures for $\eps \to 0$. Note that the integrand and the measure
  are nonnegative by optimality of $\Lam$ in $\DD$ for
  $L\mapsto \theta_L(\alpha)$ for any $\alpha>0$ (see \cite{Mont}) and since $\tilde{\mu}_\eps \geq 0$.  Hence, there is $\delta > 0$ (depending
  on $f$ and $\mu$ but independent of $\eps$) such that
  \begin{align}
    E_{\t h_\eps}[L] - E_{\t h_\eps}[\Lam] \geq \frac 12 \int_\delta^{\frac 1\delta} \big[\theta_L(\pi t) - \theta_{\Lam}(\pi t) \big] d\mu_{\hat f}(t).
  \end{align}
  By Proposition \ref{prop-locmin} and continuity of
  $(t,L)\mapsto \theta_{L}(\pi t)$, for any $\delta > 0$ we have
  $\theta_L(\pi t) - \theta_{\Lam}(\pi t) \geq c d(L,\Lam)^2$ for all
  $t \in [\delta,\frac 1\delta]$ and $L\in B_{c_1}(\Lam)$. Then \eqref{cl-2}
  follows from the estimate
  \begin{align}
    E_{\t h_\eps}[L] - E_{\t h_\eps}[\Lam] %
    \geq C d(L,\Lam)^2 \int_{\delta}^{\frac 1\delta} d \mu_{\hat f}(t)
    \geq C d(L,\Lam)^2.
  \end{align}
\end{proof}
\begin{remark}[Higher dimension]\label{rem-highdim}
  The two important ingredients here are the strict local minimality of $\Lam$
  for $\mathcal{E}_{h_\eps}$ combined with the global optimality of $\Lam$ for
  the theta function $L\mapsto \theta_L(\alpha)$ for any $\alpha>0$. The result
  of Proposition \ref{prp-localmin} can be easily generalized to dimensions
  $d\in \{4,8,24\}$, for $D_4, E_8$ and the Leech lattice, using the strict
  local minimality of these lattices for the $d$-dimensional theta function
  $L\mapsto \theta_L(\alpha)$, for any $\alpha>0$, proved in
  \cite[Thm. 4.6]{Coulangeon:2010uq} and the fact that the Fourier transform of
  a radial measure on $\R^d$ is (see e.g. \cite[Section 2]{ConnetSchwartz})
  \begin{align}
    \hat{\mu}(p)=\frac{2^{\frac{d}{2}-1}\Gamma\left(\frac{d}{2}\right)}{|p|^{\frac{d}{2}-1}}\int_0^{\infty}J_{\frac{d}{2}-1}(2\pi s |p|)s^{1-\frac{d}{2}}d\psi(s),
  \end{align}
  where $\psi$ is again the Lebesgue-Stieltjes measure of $t\mapsto
  \mu(B_t)$. In three dimensions, the problem is more involved since the local
  minimizers of the lattice theta function among Bravais lattices of unit
  density then depend on $\alpha$ (see \cite[Thm. 1.7]{BeterminPetrache}). Furthermore, some very recent results \cite{Viazovska,CKMRV} about the
  best packing in dimensions $8$ and $24$ have shown the efficiency of
  Cohn-Elkies linear programming bounds for sphere packing \cite{CohnElkies}. According to
  \cite{CohnKumar}, the next step should be the proof of the universality of
  $E_8$ and the Leech lattice, i.e. their global optimality among periodic
  lattices (not only Bravais lattices) for $L\mapsto \theta_L(\alpha)$, for any
  $\alpha>0$. Consequently, Theorem \ref{thm-1} would be true in these
  dimensions, as well as Theorem \ref{thm-2}.
\end{remark}
\begin{remark}[Numerical observations for large $\eps$]\label{rem-epslarge} %
  Our result in Theorem \ref{thm-1} is concerned with the regime of sufficiently
  concentrated masses. We note that Lemma \ref{lem-critictriang} combining with
  the fact that $h_\eps$ is radial also shows that radial mass distributions on
  $\Lam$ are critical points of the energy for any $\eps > 0$.  For large values
  of $\eps$, numerical experiments suggest that we have a periodic alternation
  of local maximality and local minimality in terms of $\eps > 0$ as illustrated
  in Figure \ref{fig-Teps} for the particular case where
    $f(x)=e^{-\pi |x|^2}$ and $\mu_\eps=\eps^{-2}\chi_{B_\eps}$ on $B_\eps$. The plot shows the numerically computed value of
  $T_{h_\eps} = D^2 \EE_{h_\eps}[\Lam]$ as a function of $\eps$. We also note
  $\Lam$ is a local minimizer for $0\leq\eps \leq \eps_0$, with
  $\eps_0\approx 0.55$.
\end{remark}
\begin{figure}[!h]
  \centering
  \includegraphics[width=10cm]{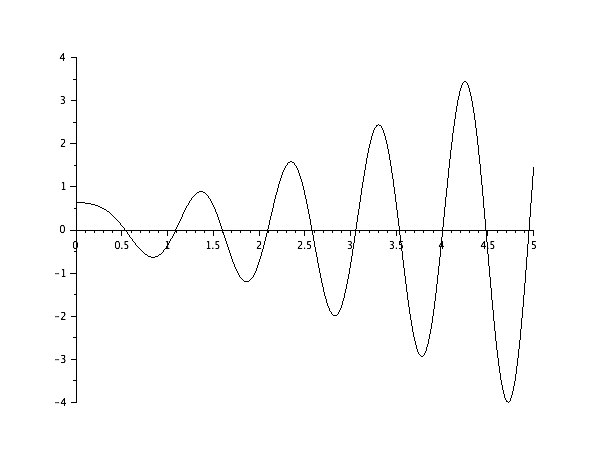}
  \caption{The plot shows the values of $T_{h_\eps}$ for
    $\mu_\eps := \eps^{-2} \chi_{B_\eps}$ and $f =e^{-\pi |x|^2}$ and for
    $\eps \in (0, 5)$, indicating if the triangular lattice is stable
    ($T_{h_\eps} > 0$) or unstable.}
  \label{fig-Teps}
\end{figure}

\subsection{Proof of Theorem \ref{thm-2}}\label{subsec-proofth1}

Let $f\in \FF$ and let $\mu$ be written as $d\mu(x) = \rho(x) dx$ for some
$\rho$ satisfying \eqref{first-con2}. With a similar argument as in the proof of Proposition
\ref{prp-epsfou}, proving that $\EE_{f,\mu}$ is minimized by $\Lam$ in $\DD$
is equivalent with proving the same for
\begin{align}
  E_{h}[L]:=\sum_{p\in L}{\vphantom{\sum}}' \hat{f}(p) \hat{\rho}^2(p).
\end{align}
  
By Lemma \ref{lem-ffhat} $\hat{f},\hat{\rho} \in \FF$. Hence, $\Phi,G$ defined by $\Phi(|p|^2) := \hat f(p)$ and
$G(|p|^2) := \hat \rho(p)$ are completely monotone. Since the product of two
completely monotone functions is again completely monotone, the function
$H =\Phi G^2$ is completely monotone. It follows (see e.g. \cite[Prop
3.1]{BetTheta15}) that the triangular lattice $\Lam$ is the unique minimizer in
$\DD$, up to rotation, of
\begin{align}
  E_{h}[L] =\sum_{p\in L}{\vphantom{\sum}}' H(|p|^2) = \sum_{p\in L}{\vphantom{\sum}}' \hat{f}(p) \hat{\rho}(p)^2.
\end{align}


\paragraph{Acknowledgement.} LB is grateful for the support of the Mathematics
Center Heidelberg (MATCH) during his stay in Heidelberg. He also acknowledges support
from ERC advanced grant Mathematics of the Structure of Matter (project
No. 321029) and from VILLUM FONDEN via the QMATH Centre of Excellence (grant
No. 10059). HK is grateful about support from DFG grant 392124319. The
authors also thank the referees for their interesting suggestions and comments.

\renewcommand{\em}[1]{\it{#1}}

 {\small
   \bibliographystyle{plain} %
   \bibliography{Diffusecharges}}




\end{document}